\documentclass{egpubl}
\usepackage{pg2019}
\usepackage{amsmath, amssymb}
\usepackage{gensymb}
\usepackage[ruled,vlined]{algorithm2e}

\DeclareMathOperator*{\argmax}{arg\,max}

\newtheorem{theorem}{Theorem}[section]
\newtheorem{lemma}[theorem]{Lemma}

 \JournalSubmission    % uncomment for submission to Computer Graphics Forum
\usepackage[T1]{fontenc}
\usepackage{dfadobe}  

\usepackage{cite}  % comment out for biblatex with backend=biber
\BibtexOrBiblatex
\electronicVersion
\PrintedOrElectronic
\ifpdf \usepackage[pdftex]{graphicx} \pdfcompresslevel=9
\else \usepackage[dvips]{graphicx} \fi

\usepackage{egweblnk} 
\begin{document}

\title[Towards Robust Direction Invariance]{Towards Robust Direction Invariance in Character Animation}

% for anonymous conference submission please enter your SUBMISSION ID
% instead of the author's name (and leave the affiliation blank) !!
% for final version: please provide your *own* ORCID in the brackets following \orcid; see https://orcid.org/ for more details.

% \author[]
% {ID: Paper1074
% }

\author[L. Ma]{
  \parbox{\textwidth}{\centering 
  Li-Ke Ma\thanks{\texttt{milkpku@gmail.com}}$^{1,2}$  ~
  Zeshi Yang\thanks{\texttt{\{zeshiy,kkyin\}@sfu.ca}}$^1$ ~
  Baining Guo\thanks{\texttt{bainguo@microsoft.com}}$^3$ ~
  KangKang Yin$^1$ 
  }
  \\
  \parbox{\textwidth}{\centering
  $^1$Simon Fraser University ~
  $^2$Tsinghua University ~ 
  $^3$Microsoft Research Asia}
}

\definecolor{myOrange}{rgb}{1,0.3,0}
\definecolor{myGreen}{rgb}{0,0.8,0.5}
\definecolor{kkBlue}{rgb}{0,0.0,0.9}

\newcommand{\like}[1]{\textcolor{myOrange}{#1}}
\newcommand{\zeshi}[1]{\textcolor{myGreen}{#1}}
\newcommand{\kk}[1]{\textcolor{kkBlue}{#1}}
\maketitle

% uncomment for using teaser
% \teaser{
%  \includegraphics[width=\linewidth]{eg_new}
%  \centering
%   \caption{New EG Logo}
% \label{fig:teaser}
%}

\begin{abstract}
%   \like{Li-Ke's color}, %\zeshi{Zeshi's color}, %\kk{KangKang's color}
In character animation, direction invariance is a desirable property. That is, a pose facing north and the same pose facing south are considered the same; a character that can walk to the north is expected to be able to walk to the south in a similar style. To achieve such direction invariance, the current practice is to remove the facing direction's rotation around the vertical axis before further processing. Such a scheme, however, is not robust for rotational behaviors in the sagittal plane. In search of a smooth scheme to achieve direction invariance, we prove that in general a singularity free scheme does not exist. We further connect the problem with the hairy ball theorem, which is better-known to the graphics community. Due to the nonexistence of a singularity free scheme, a general solution does not exist and we propose a remedy by using a properly-chosen motion direction that can avoid singularities for specific motions at hand. We perform comparative studies using two deep-learning based methods, one builds kinematic motion representations and the other learns physics-based controls. The results show that with our robust direction invariant features, both methods can achieve better results in terms of learning speed and/or final quality. We hope this paper can not only boost performance for character animation methods, but also help related communities currently not fully aware of the direction invariance problem to achieve more robust results. 

%  ACM CCS 2012
%   (see https://www.acm.org/publications/class-2012)
%The tool at \url{http://dl.acm.org/ccs.cfm} can be used to generate
% CCS codes.
%Example:


\end{abstract}  
\section{Introduction}

In computer animation, digital characters are usually parameterized by an articulated rigid body system with six Degrees of Freedom (DoFs) at the root, and a tree of internal rotational joints. It is relatively well-known, however, that the six DoFs at the root, although necessary to fully specify a 3D pose for the character, are redundant in terms of specifying the motion tasks. For example, a walk to the north and a walk in the same style to the south are considered the same walk even though the directions of motion are different. It is thus a common practice to remove the root rotation about the gravitational direction, usually the $Y$-axis in computer graphics, before further processing such as pose similarity computation in motion retrieval or control learning in physics-based animation.

To remove the redundant $Y$-rotational DoF at the root, the current practice is to define a ``facing direction'' by the $X$ axis of a local frame defined at the root, i.e., the sagittal axis, as shown in Figure~\ref{fig:RFD} left. This facing direction usually corresponds to where the character's belly button points towards. By aligning the facing directions of all frames to the same plane, poses and motion features such as positions and velocities become $Y$-rotation free and can then be compared or used in a more consistent way~\cite{Kovar04,Peng:2018:DeepMimic}.

\begin{figure}[t]
    \centering
    \includegraphics[width=\linewidth]{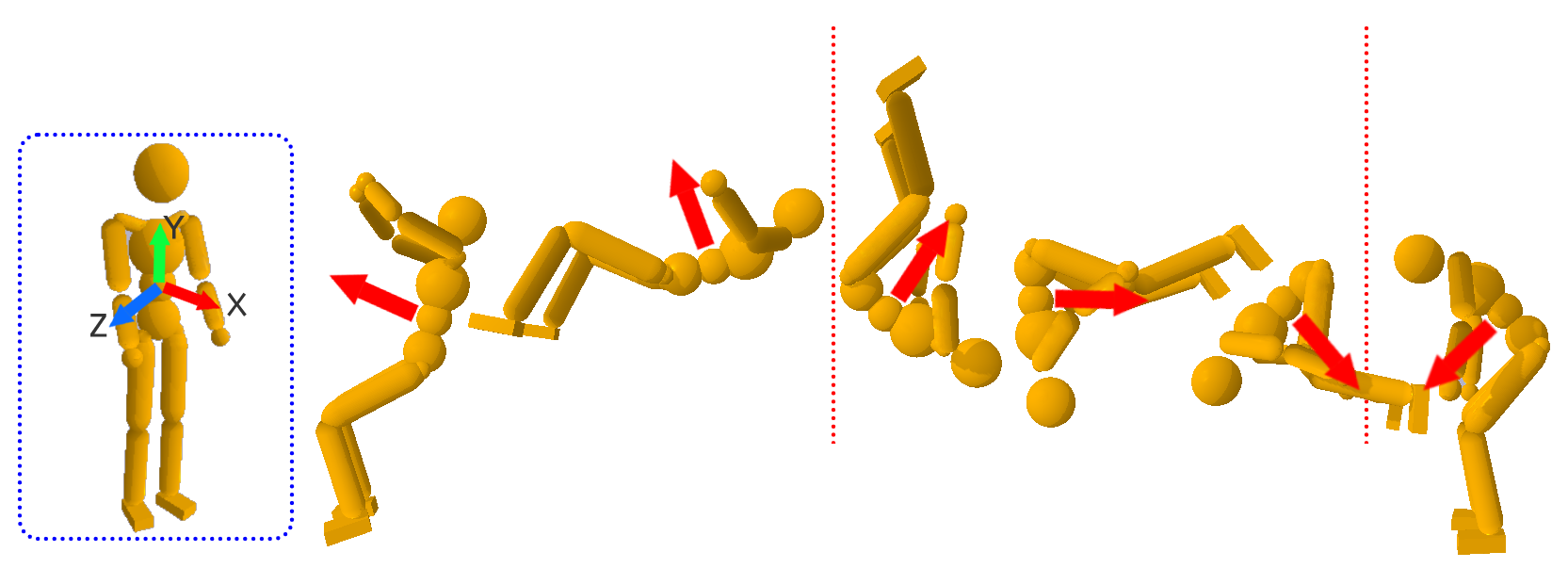}
    \caption{Left: root frame of the character model. The $X$ axis corresponds to the sagittal axis; the $Z$ axis is the lateral axis; and the $Y$ axis is the vertical axis. In human anatomy, the root $XY$ plane is referred to as the sagittal plane; the root $XZ$ plane is the transverse plane; and the root $YZ$ plane is the frontal plane. Right: Root frame facing direction ($X$ axis) during the course of a backflip.}
    \vspace{-10pt}
    \label{fig:RFD}
\end{figure}

The above defined facing direction works well for locomotion tasks, which have been the focus of most character animation research so far. There are cases, however, that the facing direction as defined above is problematic. Figure~\ref{fig:RFD} right shows a backflip motion during which the root $X$ axis aligns with the global $Y$ axis twice, around the moments denoted by the vertical dashed lines. The computed $Y$ rotations of the root frame will thus encounter singularities around these moments, as will be described in detail in Section~\ref{sec:singularity}. The motion features after being transformed by these derived $Y$-rotations will change wildly around these moments as well. Such ill-behaved features can potentially degrade the performance of various learning algorithms to certain degrees, as will be shown in Section~\ref{sec:results}. We therefore wish to devise a robust scheme that removes the direction of motion in a singularity free manner.

We formulate this problem using mathematical tools from algebraic topology in Section~\ref{sec:singularity}. Then we prove that such a desired singularity-free smooth mapping scheme actually does not exist in general. We therefore propose a remedy to choose a well-behaved ``motion direction'' to replace the ``facing direction'' for specific motions at hand. The chosen motion direction should stay away from the global $Y$ axis during the full course of the motion. For example, for the backflip shown in Figure~\ref{fig:RFD}, the root $Z$ axis qualifies as such a motion direction. We then formulate the calculation of direction-invariant features based on a given motion direction, either manually chosen or automatically computed, in Section~\ref{sec:derivation}. In Section~\ref{sec:results}, we validate that a well-chosen motion direction and its derived features are sufficient to avoid singularities in practice. We also perform comparative studies using two state-of-the-art animation methods on direction-invariant features calculated from the motion direction (DIM), direction-invariant features derived from the conventional facing direction (DIF), and global features (GF). Our results show that DIF features help achieve better performance in terms of synthesized motion quality and/or learning speed.

With recent advances and growths in machine learning and computer vision, we have seen publications related to character animation in non-graphics communities~\cite{Zhao17,zhou2018autoconditioned}. These communities seem to handle motion directions in many different ways. We show that by using our robust direction-invariant features, some of these algorithms can greatly boost their performance without any change to their training or learning components. We hope that our solution and comparative studies can help advocate robust direction-invariant features to wider communities to facilitate faster scientific advances in more related fields.

\section{Related Works}

\subsection{Character Animation}

Kinematic character animation has a long history of comparing poses~\cite{Kovar02} and modelling motions~\cite{Holden:2017:PFNN,Aristidou18} in a direction invariant manner. It is known that ``a motion is defined only up to a rigid 2D coordinate transformation. That is, the motion is fundamentally unchanged if we translate it along the floor plane or rotate it about the vertical axis''~\cite{Kovar02}. The rotation about the vertical axis can be calculated from sampled point clouds on the digital character~\cite{Kovar02}, or simply from the facing direction extracted from the root orientations~\cite{Holden16,Holden:2017:PFNN}. In this work we also use root orientations, but our framework can be applied to other schemes as well.

Various physics-based animation methods also need to handle the motion direction. In hand-crafted controllers such as \cite{Hodgins95,Yin07,Coros10}, features are transformed into the sagittal plane and the lateral plane before applying controls in each plane respectively, so that the simulated character can walk in any direction. In~\cite{Kwon17}, when the reference motion involves a high speed rotation during the flight phase, the forward facing direction is obtained by linearly interpolating the initial and final forward facing directions to achieve more robust reference coordinate systems. Trajectory optimization methods usually optimize for a specific 3D trajectory with no direction invariance~\cite{witkin88,Mordatch12,AlBorno13,levine2014learning}. Recently deep learning and deep reinforcement learning (DRL) have been applied successfully on learning physics-based controls for character animation~\cite{Peng:2018:DeepMimic,Liu18,Yu:2018:LSLL,Lee19}. Features are usually converted to be facing-direction invariant before being passed to the deep neural networks (DNNs). We chose DeepMimic~\cite{Peng:2018:DeepMimic} as one of the two methods for our comparative study in Section~\ref{sec:kinresults}, as it produces state-of-the-art results in terms of motion variety and quality, and its source code was released by the original authors and readily available.

\subsection{Related Fields}

The computer vision and machine learning communities have been investigating motion prediction and synthesis for some time~\cite{wang2007gaussian,taylor2007modeling,fragkiadaki2015recurrent,jain2016structural,martinez2017human,gui2018adversarial,zhou2018autoconditioned}. Among these publications, some avoided using root orientations at all~\cite{gui2018adversarial,zhou2018autoconditioned}; some used incremental changes around the gravitational vertical,
and absolute pitch and roll relative to the facing direction~\cite{taylor2007modeling,martinez2017human}; some just used root orientations as is~\cite{wang2007gaussian}; and some are not clear from the writing how they dealt with the root orientations~\cite{fragkiadaki2015recurrent,jain2016structural}. We chose acRNN~\cite{zhou2018autoconditioned} as one of the two methods for our comparative study in Section~\ref{sec:kinresults}, as it generates the best results in terms of long-term motion quality and complexity, and its source code was released by the original authors and readily available.

Human action recognition is also a classic problem in computer vision~\cite{lehrmann2014efficient,Simonyan:2014,Zhao17}. Most work directly deals with pixel features. A few choose to first reconstruct skeletal poses from video, and then use features in the global frame for further learning~\cite{Zhao17}. We think such recognition tasks would also benefit from using our robust direction-invariant features as well.

%In Robotics, studies have shown the importance of using ......~\cite{chen2018crowd}
%\kk{some vision-based action recognition stuff here...better to find a well-known vision paper and change the features........?where is the robotics paper you mentioned Zeshi?  }

%\section{Nonexistence of a Singularity-free Motion Direction}
\section{Nonexistence of a Singularity-free Mapping Scheme}
\label{sec:singularity}

We use $g$ to denote an orientation in $SO(3)$, and $g_1 \circ g_2$ to denote composition of two orientations $g_1$ and $g_2$. Our derivations and proofs are independent of the specific rotation parameterization method. For example, $g$ can be parameterized by quaternions and rotation composition can be achieved through quaternion multiplication. We use $g(m)$ to denote a vector $m$ rotated by $g$. 

\subsection{Problem Formulation}
We consider two orientations in $SO(3)$ equivalent if they are connected by a rotation with respect to the gravity direction, namely a rotation $g_y$ about $Y$. We denote the set of all $g_y$ as $R_Y$, which is isomorphic to $SO(2)$. We also denote the orbit of an orientation $g \in SO(3)$ under rotations in $R_Y$ as $\Omega(g)$:
\begin{equation}
  \Omega(g) = \{g_y \circ g | g_y \in R_Y\}.
\end{equation}
Rotations in the same orbit are considered equivalent. Denote the set of orbits as $H$. %$H$ is isomorphic to the 2-sphere group $S^2$, namely
Then $H$ is homeomorphic to the 2-sphere $S^2$ with a natural homogeneous space structure as the quotient of two Lie groups $SO(3)$ and $SO(2)$: 
\begin{equation}
  H = \{\Omega(g) | g \in SO(3)\} \cong SO(3)/SO(2) \cong S^2.
\end{equation}

Calculating a $Y$-rotation invariant frame for an orientation $g$ requires projecting  $\Omega(g)$ to a representative element $g_0$ of the orbit, so that all orientations in $\Omega(g)$ are aligned with $g_0$. $g_0$ is chosen manually, such as the one described in Section~\ref{sec:derivation}. Formally speaking, we wish to find a mapping scheme $\mathcal{M}$ from all orbits in $H$ to representative elements in $SO(3)$ subject to certain constraints:
\begin{equation}
  \begin{split}
    &\mathcal{M} : H \mapsto SO(3), \\
    &s.t.~\Omega\left(\mathcal{M}(h)\right) = h, ~ \forall h \in H.
  \end{split}
  \label{eq:app:setup}
\end{equation}

The question is whether there is a smooth singularity-free mapping $f$ that satisfies the above property. We provide two proofs that such singularity-free mapping does not exist. One proof follows the Fiber Bundle theory~\cite{Hatcher02}. The other connects to the Hairy Ball Theorem~\cite{Eisenberg:1979:HBT}, which is more familiar to the graphics community~\cite{Fisher07}.

\subsection{Proof 1: by Fiber Bundle Theory}
Our orbit operator $\Omega$ defines a fiber bundle $\left(SO(3), S^2, \Omega, SO(2)\right)$. In the quest of a smooth mapping $\mathcal{M}$, we wish to find a globally defined cross section of this fiber bundle. However, \begin{equation}
    SO(3) \not \cong SO(2) \times S^2,
\end{equation}
since SO(3) is homeomorphic to the 3-dimensional projective space, which is not homeomorphic to $SO(2) \times S^2$. Thus there does not exist such a global cross section~\cite{Hatcher02}.

\subsection{Proof 2: by Hairy Ball Theorem}
\label{sec:hairy_ball}

Imagine a unit tangent vector $t=(1,0,0)$ attached to the north pole $p=(0, 1, 0)$ of a unit sphere. Now we define a mapping $\mathcal{F}$ that maps an element $g\in SO(3)$ to a unit tangent vector $g^{-1}(t)$ located at $g^{-1}(p)$ on the sphere, by rotating this sphere by $g^{-1}$.

\vspace{5pt}
\begin{lemma}
The mapping $\mathcal{F}$ between $SO(3)$ and unit tangent vectors on the sphere is bijective.
\label{lemma:mapping:so3}
\end{lemma}
\begin{proof}
(a) Each element in $SO(3)$ rotates $t$ to a unique unit tangent vector on the sphere, so the mapping is injective. (b) For any unit tangent vector $s$ on the sphere, there is always a $g^{-1}$ that can rotate $t$ to $s$, so the mapping is surjective. Thus the mapping $\mathcal{F}$ is bijective.
\end{proof}

\vspace{5pt}
\begin{lemma}
\label{lemma:mapping:orbit}
  $\mathcal{F}$ is a bijective mapping between $\Omega(g)$ and the unit tangent vector space at $g^{-1}(p)$.
\end{lemma}
\begin{proof}
  (a) $\forall \tilde{g}=g_y \circ g \in \Omega(g)$, $\tilde{g}^{-1}$ sends $p$ to a fixed point $g^{-1}(p)$ on the sphere, and $t$ to a unit tangent vector at this point, because
  \begin{equation}
    \tilde{g}^{-1}(p) = (g^{-1} \circ g_y^{-1})(p) = g^{-1}(p).
  \end{equation}
  (b) Any unit tangent vector $s$ at $g^{-1}(p)$ corresponds to a rotation in $\Omega(g)$. This is because for any unit tangent vector $s$ at $g^{-1}(p)$, $s$ corresponds to a rotation $g_s \in SO(3)$ due to the bijection property of $\mathcal{F}$ that we just proved in Lemma~\ref{lemma:mapping:so3}, and $g_s^{-1}$ also sends $p$ to $g^{-1}(p)$. Now we take a look at the point $p$:
  \begin{equation*}
    p = (g_s \circ g_s^{-1})(p) = g_s(g_s^{-1}(p)) = g_s(g^{-1}(p)) = (g_s \circ g^{-1})(p),
  \end{equation*}
  $p$ is invariant under the rotation $g_s \circ g^{-1}$. We know that only $g_y \in R_Y$ can send north pole $p$ to itself, thus $g_y = g_s \circ g^{-1}$, and therefore $g_s = g_y \circ g \in \Omega(g)$.
\end{proof}

\begin{theorem}
Under the mapping $\mathcal{F}$, a solution of Equation~\ref{eq:app:setup} corresponds to a unit vector field on the sphere. 
\end{theorem}
\begin{proof}
Since there is a bijection between $\Omega(g)$ and unit tangent vectors at $g^{-1}(p)$ according to Lemma~\ref{lemma:mapping:orbit}, selecting a representative element from each orbit $\Omega(g) \in H$, as the mapping $\mathcal{M}$ in Equation~\ref{eq:app:setup} does, is equivalent to selecting one unit tangent vector for each point $g^{-1}(p)$ on the sphere. Therefore $\mathcal{M}$ corresponds to a unit tangent vector field on the sphere. Hereafter we denote this corresponding vector field as $V_{\mathcal{M}}$.
\end{proof}

This is exactly what the Hairy Ball Theorem concerns~\cite{Eisenberg:1979:HBT}: you simply cannot comb a sphere without creating at least one cowlick somewhere. That is, there is no singularity-free unit tangent vector field on a sphere, which is relatively well-known in the graphics community~\cite{Fisher07}. Therefore, there is no singularity-free mapping that satisfies Equation~\ref{eq:app:setup}.
\section{Robust Direction-invariant Features}
\label{sec:derivation}

\begin{algorithm}[t]
\label{alg:mdc}
  \SetAlgoLined
  \KwIn{Root orientation frames $\{g_t\},~t=1, \dots, T$}
  \KwOut{The motion direction vector $r$}
 
  $p \leftarrow (0, 1, 0)$\;
  \For{$t \leftarrow 1$ \KwTo $T$}{
    $p_t = g^{-1}_t(p)$\;
   }
  $\{r_i\}$=UniformlySampleOnSphere($N$), $~ i = 1, \dots, N$\;
  \For{$i \leftarrow 1$ \KwTo $N$}{
    $d_i \leftarrow +\infty$\;
    \For{$t \leftarrow 1$ \KwTo $T$}{
    $d^{+}_{it}=$ComputeGeodesicDistance($r_i,p_t$)\;
    $d^{-}_{it}=$ComputeGeodesicDistance($-r_i,p_t$)\;
    $d_i = \min\{d_i,~d^{+}_{it},~d^{-}_{it}\}$\;
    }
  }
  $j \leftarrow \argmax_{i\in 1, \dots , N} d_i$\;
  \Return $r_j$
\caption{Motion Direction Computation}
\end{algorithm}

In order to map all motion features into a direction invariant set, we need to first calculate the $Y$-rotation of the root frame $g$, so that the transformed features become invariant to the direction of motion. We formulate one of the typical algorithms used in the animation community for this purpose in Section~\ref{sec:dif}, but generalize it by replacing the facing direction with motion direction. The motion direction  $r$ can be either manually selected for simple rotational tasks, or automatically computed as described in Section~\ref{sec:mdc}.

\subsection{Motion Direction Computation}
\label{sec:mdc}

Following the proof in Section~\ref{sec:hairy_ball}, we derive an algorithm as shown in Algorithm~\ref{alg:mdc} to automatically choose the best motion direction $r$ for a given root orientation trajectory $\{g_{t}\}$. The idea is to maximize the geodesic distance from $p=(0,1,0)$ to $g_t(r)$ and $g_t(-r)$ to stay away from singularities for the full course of the motion, which is equivalent to maximizing the geodesic distance from $g_t^{-1}(p)$ to $r$ and $-r$. We will also use the trajectory $g_t^{-1}(p)$ and the vector field $V_{\mathcal{M}}$ in our visualizations in Figure~\ref{fig:vector_field} and \ref{fig:break_dance}.

\subsection{Direction-invariant Features}
\label{sec:dif}

Assuming the current root frame orientation represented in the global world frame is $g$. Its $Y$-rotation $g_y$ induced by motion direction $r$ can be derived as follows as illustrated in Figure~\ref{fig:scheme}. First, rotate the motion direction $r$ by $g$ to get $r' = g(r)$. Second, calculate the dihedral angle $\theta$ between the two planes with normals $Y \times r$ and $Y \times r'$. Then $g_y$ equals $(Y,\theta)$, which is a rotation that rotates $r'$ onto the plane $Y \times r$ and parameterized in the axis-angle representation. Finally, $g_0 = g_y \circ g$, and all orientations in orbit $\Omega(g)$ are mapped to the same representative element $g_0$. The $V_{\mathcal{M}}$ that corresponds to the mapping $\mathcal{M}$ induced from a motion direction $r$ can then be calculated as $g_0^{-1}(t)$ for point $g^{-1}(p)$ on the sphere. Note that we only
need $g_y$ and not $g_0$ to transform a motion feature $m$ into its direction invariant version $g_y(m)$. In our experiments in Section~\ref{sec:results}, motion direction $r$ is chosen as follows: by default we choose the original facing direction $X$ axis; if it does not work, we choose the lateral $Z$ axis; when both fail, we compute another axis as described in Algorithm ~\ref{alg:mdc}.

\begin{figure}
    \centering
    \includegraphics{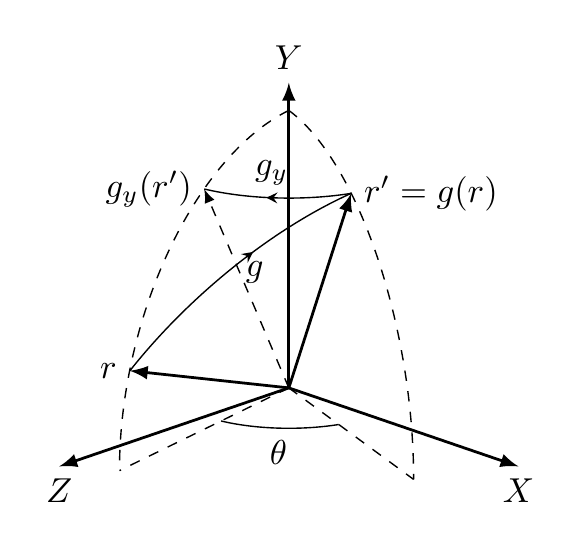}
    \vspace{-15pt}
    \caption{Calculation of $g_y$ using a chosen motion direction $r$.}
    \vspace{-15pt}
    \label{fig:scheme}
\end{figure}

As we have proved in Section~\ref{sec:hairy_ball}, one direction-invariant mapping scheme $\mathcal{M}$ corresponds to one unique unit vector field $V_{\mathcal{M}}$ on the sphere. We thus visualize two mapping schemes in Figure~\ref{fig:vector_field}. Figure~\ref{fig:vector_field} (a) is induced by choosing the conventional facing direction (root $X$ axis) as the motion direction. Figure~\ref{fig:vector_field} (b) is induced by choosing the lateral direction (root $Z$ axis) as the motion direction. We also visualize the trajectory of root frames $\{g_t\}$ by plotting $g_t^{-1}(p)$ on the unit sphere. The singularity points are calculated by Helmholtz decomposition~\cite{Crane:2013:DGP} from the vector field $V_{\mathcal{M}}$. As we can see, different mapping schemes put singularity points onto different locations on the sphere, while the trajectories of root frame orientations remain the same. Thus we can strategically avoid singularities by choosing an appropriate motion direction for a specific motion task. %For tasks such as rolling or crawling, the right direction are recommended than the conventional facing direction. 

\begin{figure}
    \centering
    \includegraphics[width=\linewidth]{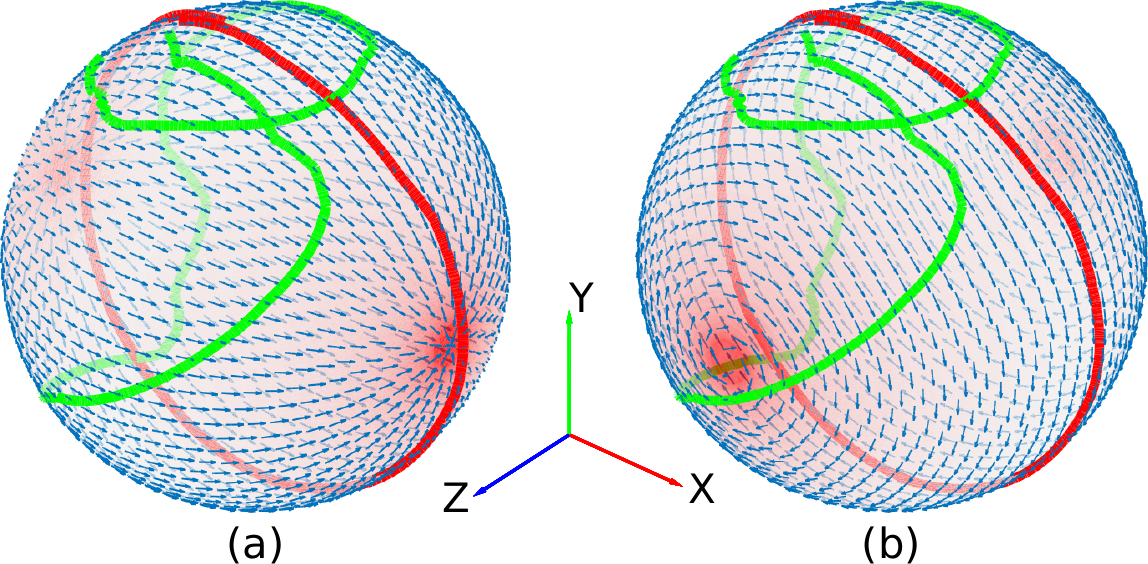}
    \caption{Visualization of two mapping schemes: (a) Facing direction (root $X$ axis) as the motion direction. (b) Lateral direction (root $Z$ axis) as the motion direction. The curves on the sphere are the mapped trajectories of the root orientations in a backflip (red) and a jump twist (green). They pass the singularities in only one of the mapping schemes but not the other.}
    \vspace{-15pt}
    \label{fig:vector_field}
\end{figure}
\section{Results}
\label{sec:results}

We perform comparative studies on motion synthesis tasks using one kinematic method~\cite{zhou2018autoconditioned} and one physics-based approach~\cite{Peng:2018:DeepMimic}. In both studies, we compare performance using global features (GF), direction-invariant features calculated from the facing direction (DIF), and direction-invariant features calculated from the motion direction (DIM). Note that DIF is a special case of DIM. So we first compare GF vs. DIF, and if DIF performs better than GF, DIM will too. Then we compare DIF vs. DIM for those motions that DIF is problematic. In the paper, we report the performance in terms of motion quality and learning speed. We encourage the readers to also watch the supplemental video demo to better judge the differences visually. We use the same color scheme for all the figures and video clips: the green character shows the ground truth captured reference motion; the red character visualizes results from GF; the blue character shows motions synthesized from DIF; and the yellow character presents motions from DIM. 

%\subsubsection{\kk{Motion Recognition}}

%cite?? wang2007gaussian,taylor2007modeling,
\subsection{Kinematic Character Animation}
\label{sec:kinresults}
We chose acRNN~\cite{zhou2018autoconditioned} as the kinematic motion synthesis method to perform the comparative study, as it is the only one that can synthesize long-term and complex motions, to the best of our knowledge. Trained acRNN (auto-conditioned Recurrent Neural Net) can generate long motions of style similar to those in the training set given an input short motion clip. We use the code released by the original authors and all training settings and parameters remain the same, except for the input motion features which we use three different sets: GF, DIF, and DIM. The features used in the original paper, including the root velocity and relative positions of internal joints wrt. the root, are represented in the global world frame and thus used as the GF set. We further convert all features to the DIF and DIM sets, with an addition of the root's angular velocity around the $Y$ axis to generate the full motion trajectory by integration. All models with different input feature sets are trained for the same amount of time. We use the publicly available CMU motion capture database and downsample the selected motions to 60~Hz as our training data.

\begin{figure}[]
   % an empty figure just consisting of the caption lines
   \centering
   \includegraphics[width=1\linewidth]{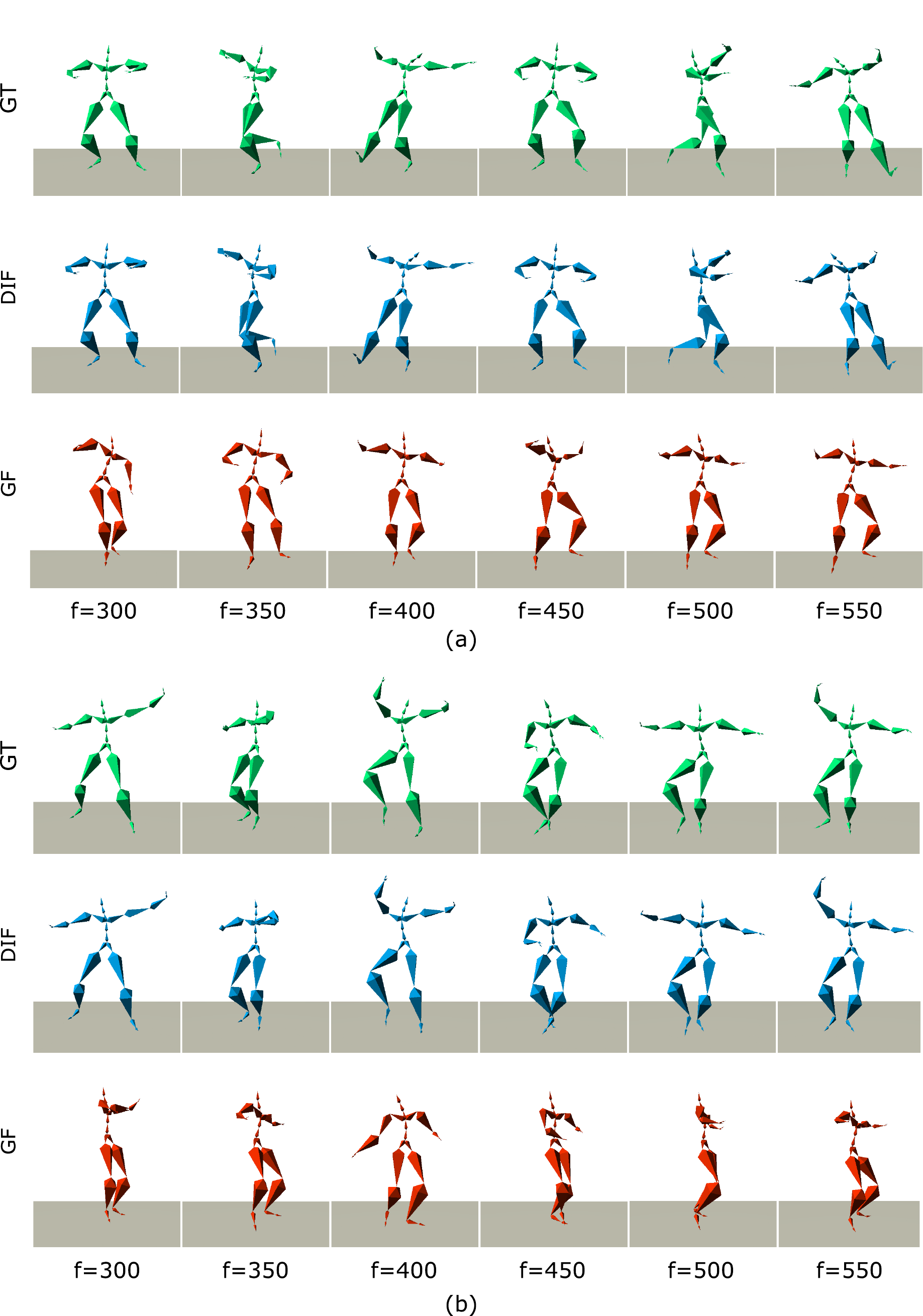}
   \caption{\label{fig:AcLSTM-global}
   Comparison of motion synthesis models trained with GF vs. DIF. f is the frame index. The green character on the top row shows the ground truth motion capture data (GT). The blue character in the middle row shows poses generated by the DIF-trained model. The red character on the bottom row shows poses generated by the GF-trained model.}
   \vspace{-15pt}
\end{figure}

\begin{figure}[]
   % an empty figure just consisting of the caption lines
   \centering
   \includegraphics[width=1\linewidth]{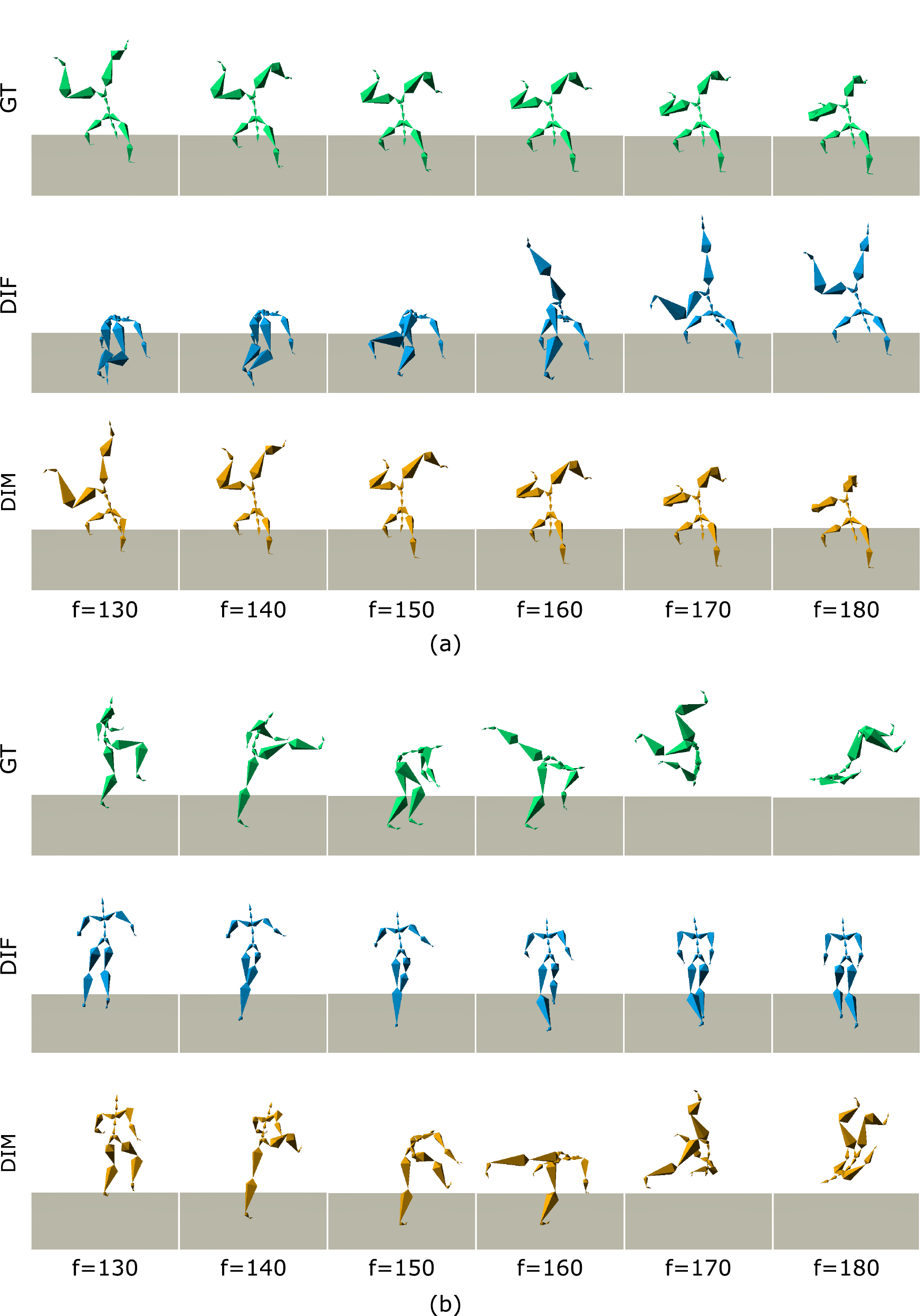}
   \caption{\label{fig:AcLSTM-xaixs} Comparison of motion synthesis models trained with DIF vs. DIM. The green character on the top row shows the ground truth motion capture data (GT). The blue character in the middle row shows poses generated by the DIF-trained model. The yellow character on the bottom row shows motions generated by the DIM-trained model.}
   \vspace{-15pt}
\end{figure}

\subsubsection{GF vs. DIF}
\label{sec:kin:global}
To compare models with GF vs. DIF, we choose the original Indian Dance Motion Sets used in the original paper~\cite{zhou2018autoconditioned}. This dataset contains roughly 421 seconds of training motion, and we trained both models for 50 hours. Figure~\ref{fig:AcLSTM-global} shows several snapshots taken from synthesized motions generated by the trained models. The model trained with rotation-invariant features DIF can generate significantly longer motions in higher qualities, in contrast to motions generated by GF, which tend to become unrealistic after a few seconds. The rotation-invariant feature representation greatly reduces the complexity of the models that need to be learned.
   
\subsubsection{DIF vs. DIM}
\label{sec:kin:zaxis}
To compare models learned with DIF vs. DIM features, we prepared a dataset of about 203 seconds which consists of various rotational behaviors such as break dance and acrobatics motions. We trained both models for 40 hours. Figure \ref{fig:AcLSTM-xaixs} shows snapshots of synthesized motions. The model trained with DIF fails to generate some rotational motions, while the DIM model using the root $Z$ axis as the motion direction can successfully reproduce these behaviors.

\begin{figure}[]
   % an empty figure just consisting of the caption lines
   \centering
   \includegraphics[width=1\linewidth]{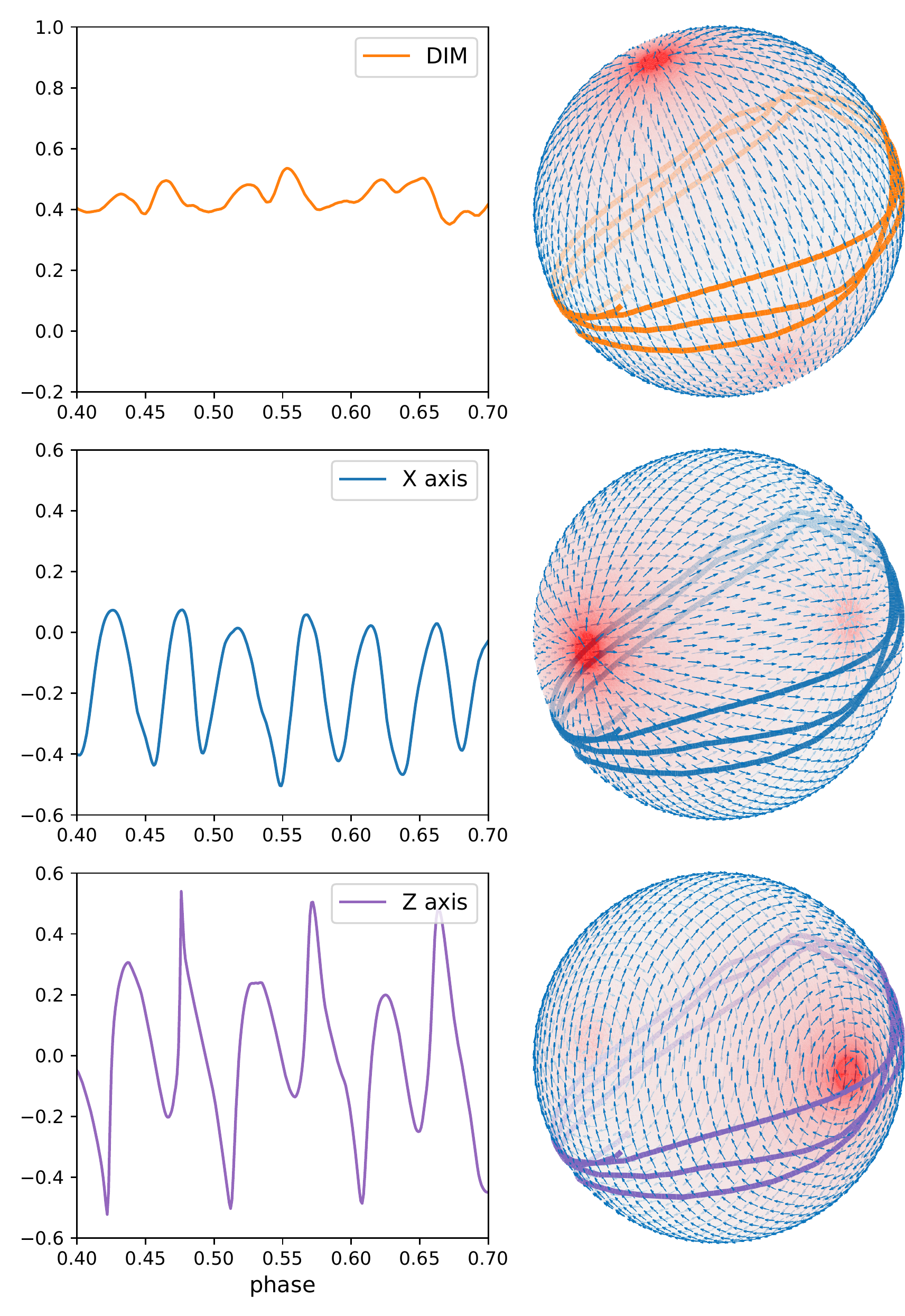}
   \caption{\label{fig:break_dance}
   Comparison of motion features derived from DIM, DIF (the $X$ axis) and the $Z$ axis for a spin dance motion. The specific feature visualized here is the $z$ coordinate of the relative position from the chest to the root. phase is the current time divided by the total length of the motion.}
      \vspace{-10pt}
\end{figure}

We further note that when the facing direction root $X$ axis does not work as a valid motion direction, the next direction we try is the lateral direction root $Z$ axis. However, there are rotational behaviors that neither direction qualifies as a good motion direction. Figure~\ref{fig:break_dance} shows one example with a spin dance motion. Features derived from both the $X$ axis and the $Z$ axis are bad. We thus use Algorithm~\ref{alg:mdc} to automatically compute a robust motion direction for this motion:  $(-\sin(30\degree),\cos(30\degree),0)$. As shown at the top of  Figure~\ref{fig:break_dance}, features derived from this axis are much smoother. In the supplementary video, we also visualize all the transformed positional features, and those derived from the $X$ axis and $Z$ axis will cause the character to rotate wildly. 

\subsection{Physics-based Character Animation}

We chose DeepMimic~\cite{Peng:2018:DeepMimic} as the physics-based motion synthesis method to perform the comparative study, as it produces the best results in terms of motion quality and motion variety. DeepMimic combines a motion-imitation objective with a task objective, and train characters using DRL to, say, walk in a
desired direction. The input features are of high dimensions, and we only alter the state that describes the configuration of the character's body, including relative positions of each link with respect
to the root, their rotations expressed in
quaternions, and their linear and angular velocities. We refer interested readers to the original paper for more details~~\cite{Peng:2018:DeepMimic}. We run the code released by the original authors on a machine with an 18-core i9-7980XE CPU, and can draw about five million samples per hour. We keep all other training parameters the same. The training motion capture clips are either from the data downloaded together with the code, or selected from the CMU mocap database.

\begin{figure}[t]
    \centering
    \includegraphics[width=\linewidth]{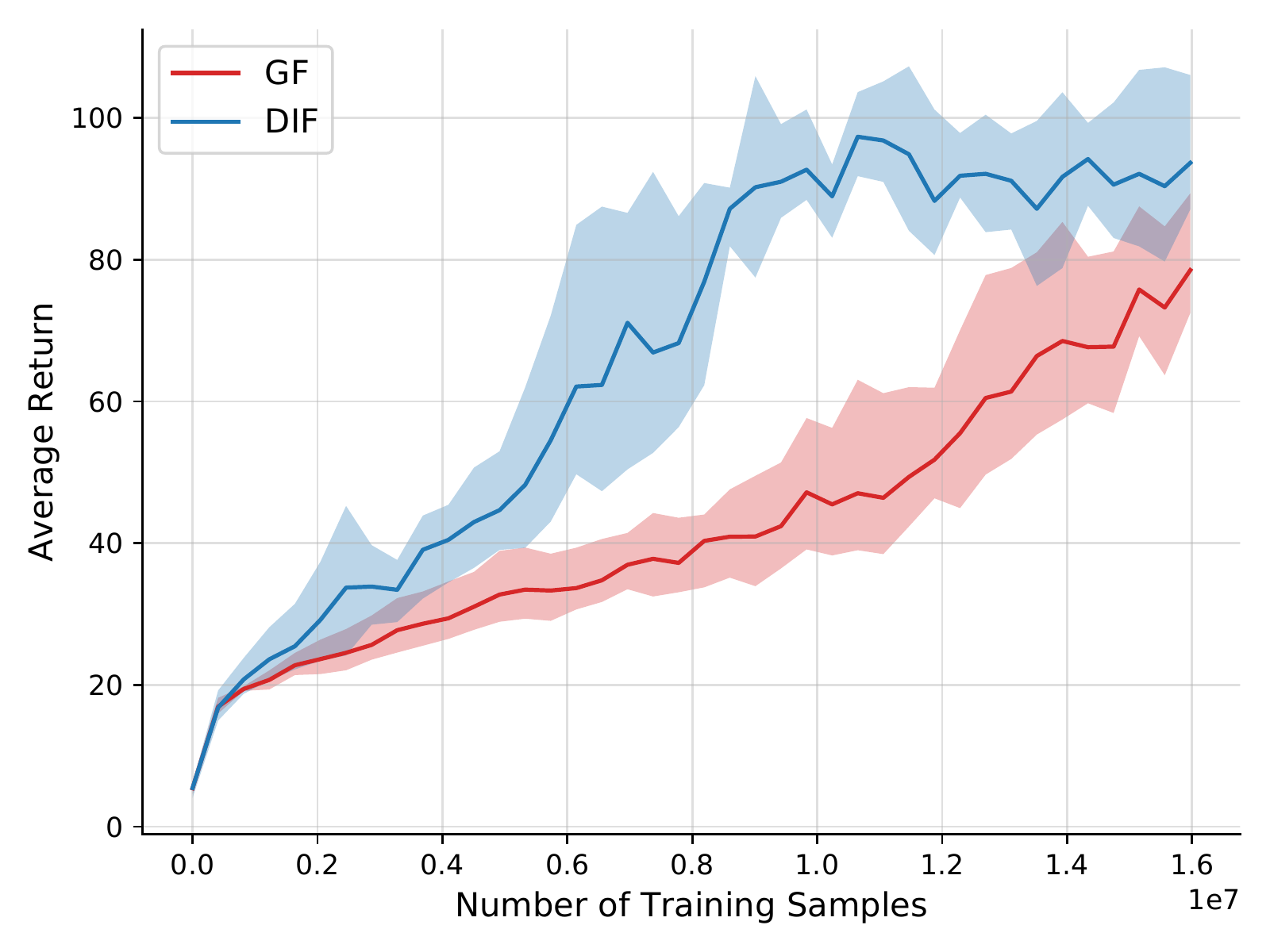}
    \caption{Training curves of a walk. The control policy trained using GF cannot walk without falling while the policy trained using DIF can walk indefinitely.}
    \label{fig:curve:walk}
\end{figure}

\begin{figure}[t]
    \centering
    \includegraphics[width=\linewidth]{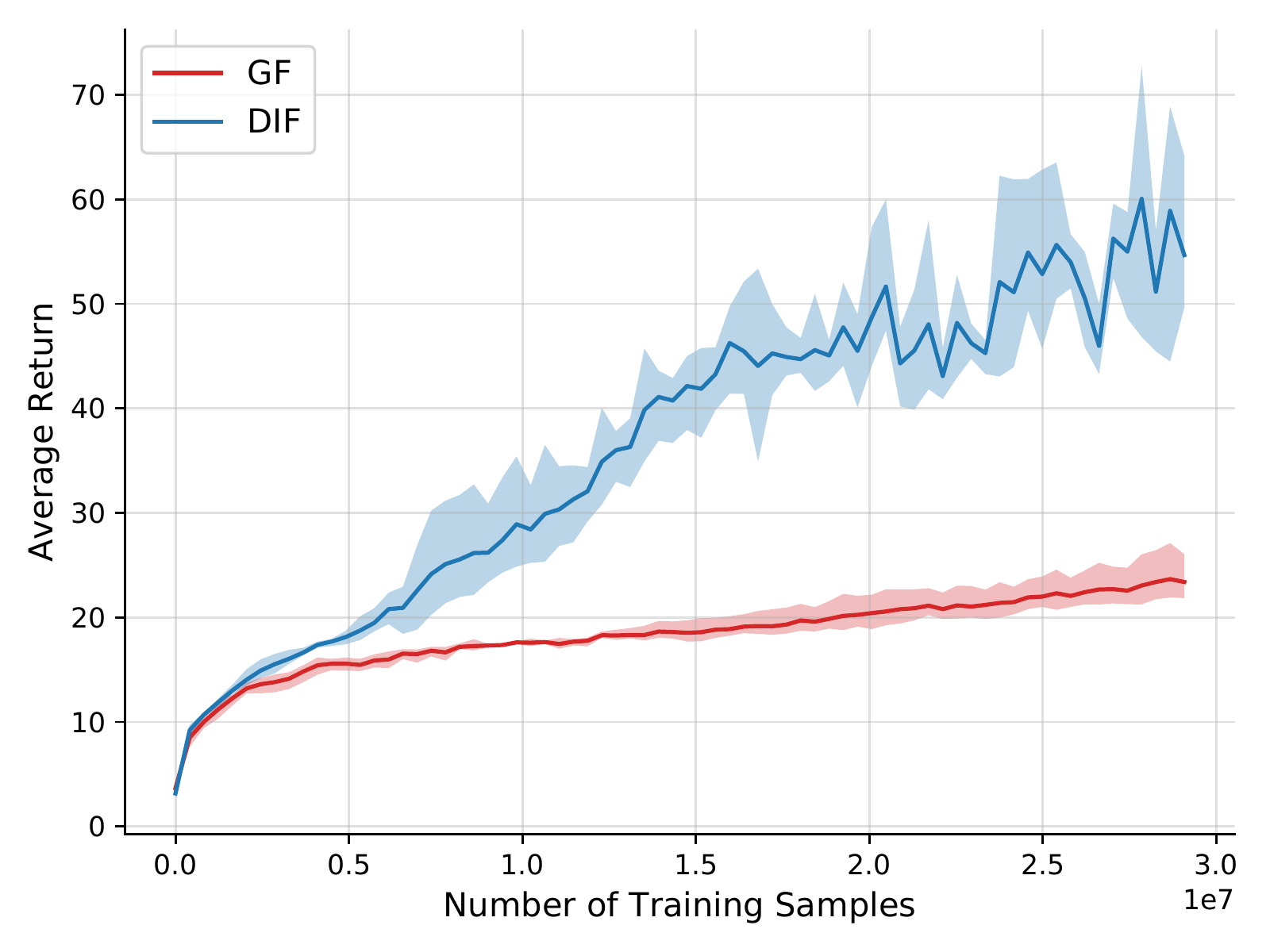}
    \caption{Training curves of a backflip. The control policy trained using GF cannot reproduce the backflip while the policy trained using DIF can backflip indefinitely.}
    \label{fig:curve:backflip}
    \vspace{-10pt}
\end{figure}

\subsubsection{GF vs. DIF}
\label{sec:dynamic:global}
%\kk{change DeepMimic's local features to global features to see if the performance degrades}
To compare controllers learned from GF vs. DIF, we used two cyclic motion tasks that slowly turn in the $XZ$ plane: a walk ($\sim$8.8s) and a backflip ($\sim$10.6s). The original DeepMimic code used features in the facing frame, i.e., the DIF scheme. We further tested features in the global frame GF. As there is some randomness in deep reinforcement learning, we trained each task three times with different random seeds, and evaluated the policy every 100th iteration as in DeepMimic. Figure~\ref{fig:curve:walk} and \ref{fig:curve:backflip} show the training curves for the walk and backflip, respectively. Solid curves correspond to mean returns and shaded regions correspond to minimal and maximal returns among the three trials. Using GF largely slows down the training and results in failure in both tasks, i.e., the character loses balance and falls to the ground. In contrast, the DIF features facilitated learning of direction-invariant locomotion controllers. We also encourage readers to refer to our supplement video to see the animated results. Note that there is little difference in training of a short straight walk though, for example. In such cases the controller does not need to be invariant to different facing directions, and both schemes can be successful.

\subsubsection{DIF vs. DIM}
\label{sec:dynamic:zaxis}
To compare DIF vs. DIM, we tested rotational motions such as crawl, roll and backflip. We first plot and compare the DIF and DIM features for a backflip in Figure~\ref{fig:rep:backflip}. There are severe discontinuities in the motion features in the DIF scheme. We also trained DRL policies with DIF and DIM features. The DIM motion direction is chosen to be the root $Z$ axis in all test cases. Learning with  DIM features results in faster training and better motion quality as shown in Figure~\ref{fig:crawl}. In these tasks, the mapped root frame crossed the singularities on the unit sphere in the DIF scheme (e.g., Figure~\ref{fig:scheme}), while the DIM scheme successfully avoided them. Note that for certain cases, such as backflip, learning with DIF and DIM features presented little difference. We conjecture that the controller is easy to be successful at the moment of singularity crossing, due to the inertia of motion in the midair.

\begin{figure}
    \centering
    \includegraphics[width=\linewidth]{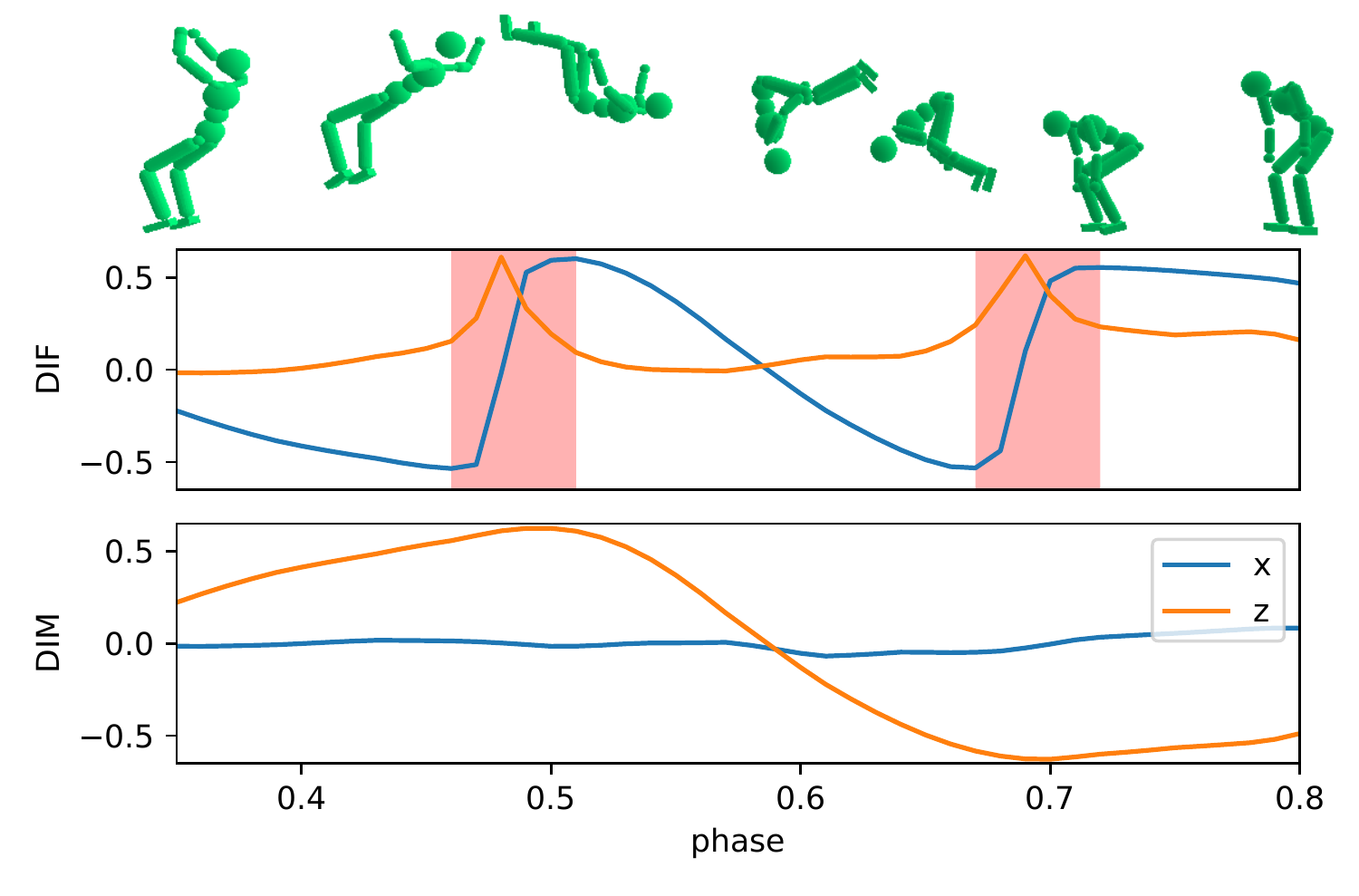}
    \caption{Comparison of DIF and DIM features for a backflip. The $x$ and $z$ coordinates of the relative position of the neck to the root are visualized. There are two moments where the DIF features are discontinuous while the DIM features are always smooth.}
    \vspace{-10pt}
    \label{fig:rep:backflip}
\end{figure}

% \begin{figure}
%     \centering
%     \includegraphics[width=\linewidth]{img/crawl_feature.pdf}
%     \caption{Feature of crawling represented in DIF and DIM.}
%     \label{fig:rep:crawl}
% \end{figure}

% \begin{figure}
%     \centering
%     \includegraphics[width=\linewidth]{img/roll_feature.pdf}
%     \caption{Feature of rolling represented in DIF and DIM.}
%     \label{fig:rep:roll}
% \end{figure}

%\like{We train policy using DeepMimic with input state represented in GF, DIF and DIM on craw and roll. The motion direction of DIM is chosen to be local $Z$ axis. Both of these motions share the similarity that their root frame orientations cross the singularity point of DIF at some moments. The results show that using DIM can avoid singularities over the full course of motion (Figure \ref{fig:rep:crawl} and \ref{fig:rep:roll}), and can result in a faster training procedure (Figure \ref{fig:curve:crawl} and \ref{fig:curve:roll}) and better visual results (Figure \ref{fig:compaire:crawl} and \ref{fig:compaire:roll}).}

\begin{figure}
    \centering
    \includegraphics[width=\linewidth]{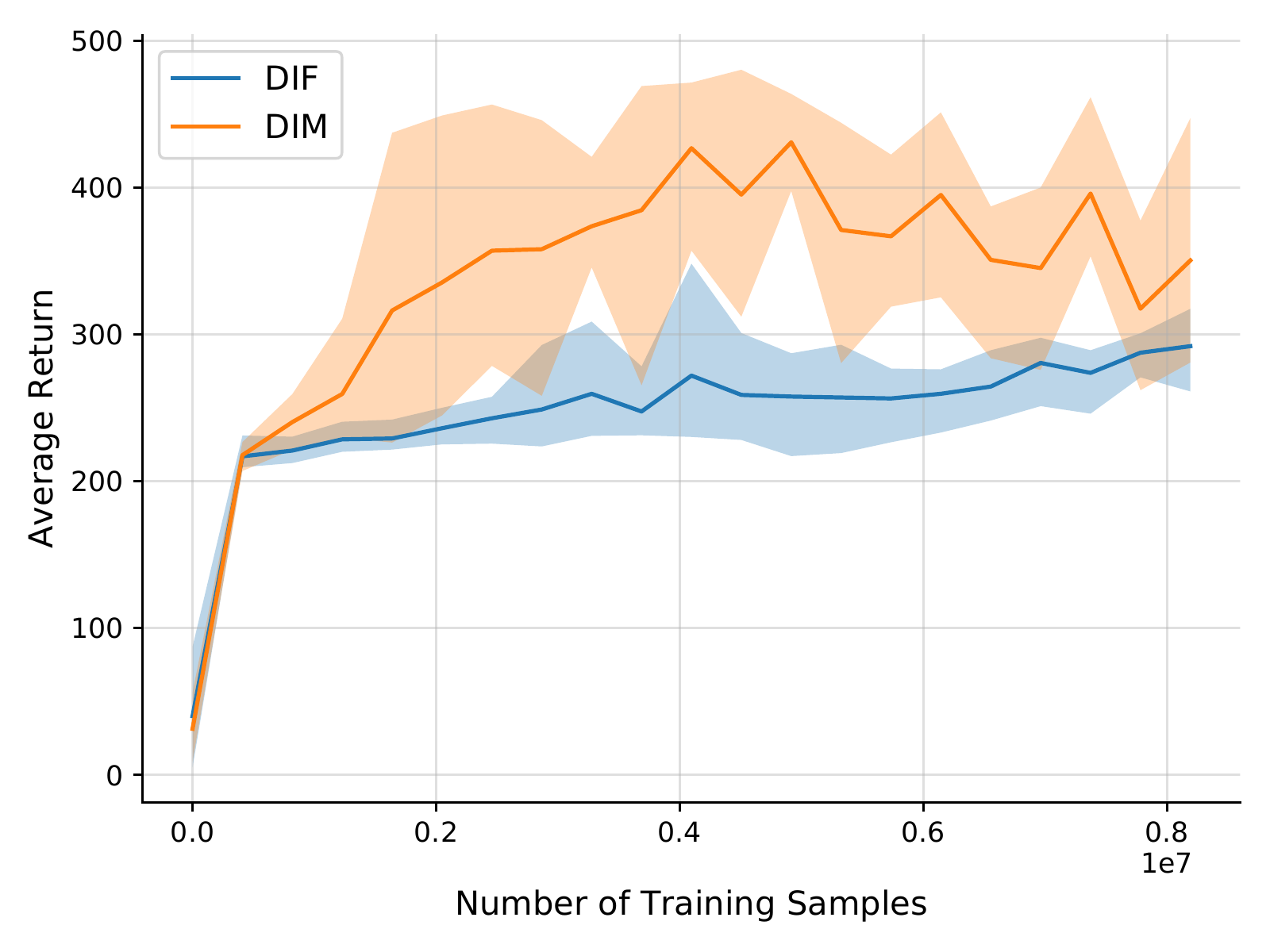}
     \includegraphics[width=0.9\linewidth]{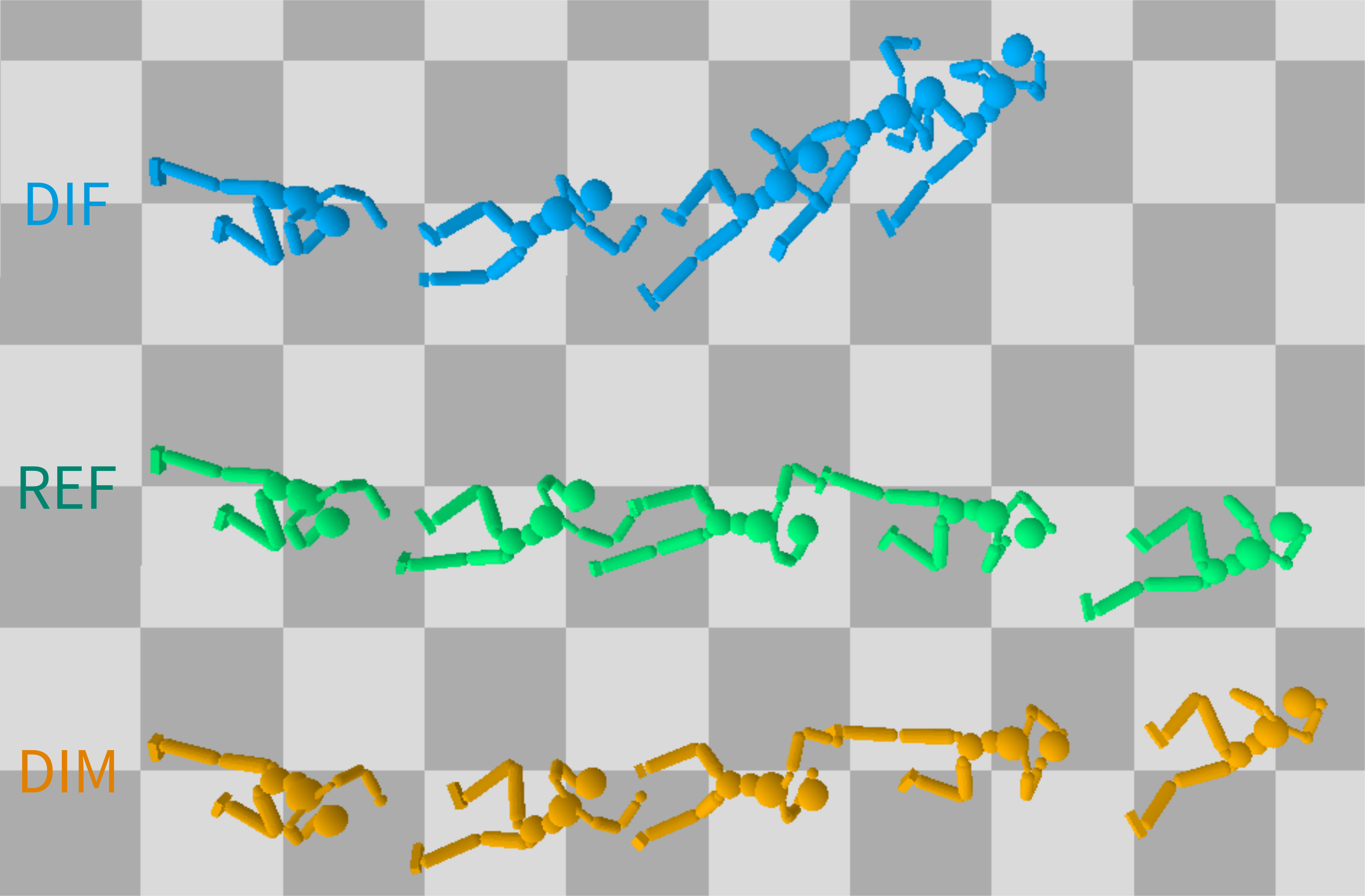}
    \caption{Top: training curves of a crawl. Training using DIM features gain higher reward faster. Bottom: the policies after 2000 training iterations ($\sim$ eight million samples) produce motions of different quality. The policy trained from DIF features cannot crawl in a straight line and may get stuck at certain moments.}
    \label{fig:crawl}
       \vspace{-15pt}
\end{figure}
% \begin{figure}
%     \centering
%     \includegraphics[width=\linewidth]{img/roll_curve.pdf}
%     \caption{Training Curve of Roll.}
%     \label{fig:curve:roll}
% \end{figure}

% \begin{figure}
%     \centering
%     \includegraphics[width=\linewidth]{img/roll_compare.pdf}
%     \caption{\like{Visual Results of rolling Policies training using DIF and DIM after 2000 iterations. Green character stands for reference, blue for DIF, and yellow for DIM. The $X$ coordinate is reset linearly increasing with the frames so that offset along $Z$ axis can be faithfully compared.}}
%     \label{fig:compare:roll}
% \end{figure}

\section{Discussion}

To conclude, we have proved that in general there is no singularity-free scheme to achieve direction invariance in character animation. However, a properly-chosen motion direction and its derived features can stay far away from singularities for specific motions and work well in practice. The traditionally-chosen facing direction has worked so far because locomotion tasks, which have been the focus of character animation research, generally only involve rotations around the gravitational direction. Our robust direction invariant features derived from motion directions will enable better results for modelling and synthesis of more versatile behaviors such as gymnastic and dance motions. A good motion direction can be either selected manually, or computed automatically for complex rotational behaviors that contain rotations around multiple axes. 
%, whether controlled by a kinematic or physics-based controller

We note that even though kinematic character controllers are usually facing direction invariant today, most physics-based controllers are not. For example in DeepMimic, the learned controller tries to track the absolute root orientation in the reference motion. If the reference character faces north and walks north, the simulated character will always try to face north. Even if the simulated character is initially positioned facing south, it will still try to turn north and walk north. This will result in either a fall or an awkward sharp turn to the north. If this global orientation constraint could be lifted in the DRL learning, we would expect greater performance gaps among GF, DIF and DIM schemes. Optimization-based control methods such as trajectory optimization also have mainly worked on direction-specific controllers. We expect that our DIM features will achieve performance gains over DIF as well for optimization-based methods, as they usually need to calculate derivatives, and DIM features are simply smoother and more well-behaved than DIF. Lastly, we note that although transformed rotations in $SO(3)$ are smooth using our DIM scheme, widely used rotation representations such as quaternions can still have discontinuities and pose difficulty for learning. We recommend using continuous rotation representations such as the ones recently proposed in \cite{Zhou_2019_CVPR}.

In the future, we wish to examine how to achieve direction invariance for end-to-end computer vision systems such as for human activity recognition. In such systems, only pixel data and/or 2D positions are available as input features, and the DNN needs to learn direction invariance somehow. 3D direction-invariant features will definitely accelerate this process, but it is not obvious how to derive such features from 2D inputs without another layer of 2D to 3D pose lifting~\cite{Tome17}.

%\kk{Discuss our solution scheme vs. a general scheme. Generality of the proofs.}

\paragraph*{Acknowledgements}
%    We are grateful to the following for resources, discussions and suggestions. 
We would like to thank Runjie Hu for insightful discussions during the development of our proofs in Section 3, and Yiying Tong for proof reading these proofs. This project is partially supported by NSERC Discovery Grants Program RGPIN-06797 and RGPAS-522723. Li-Ke Ma is supported by a CSC scholarship from the China Scholarship Council.

\bibliographystyle{eg-alpha-doi} 
\bibliography{hairyball.bib}

%\appendix
%\input{src/Scheme.tex}

\end{document}